\documentclass[10pt]{article}
\usepackage{amssymb, amsmath}
\usepackage{fullpage}
\usepackage{subfig}
\usepackage{graphicx,amsmath,amsfonts,amscd,amssymb,bm,cite,epsfig,epsf,url,color}
\usepackage[bookmarks=false]{hyperref}
\usepackage{caption}
\usepackage{array}

\numberwithin{equation}{section}

\def \endprf{\hfill {\vrule height6pt width6pt depth0pt}\medskip}
\newenvironment{proof}{\noindent {\bf Proof} }{\endprf\par}
\newtheorem{theorem}{Theorem}[section]
\newtheorem{lemma}[theorem]{Lemma}
\newtheorem{corollary}[theorem]{Corollary}

\newtheorem{definition}[theorem]{Definition}

\renewcommand{\>}{\rangle}
\newcommand{\Z}{\mathbb{Z}}
\newcommand{\R}{\mathbb{R}}
\newcommand{\C}{\mathbb{C}}
\newcommand{\srf}{\text{SRF}}
\newcommand{\srfeq}{\text{\emph{SRF}}}
\newcommand{\PTn}{P_{T} }

\newcommand{\PTcn}{P_{T^{c}}}

\newcommand{\normInf}[1]{\left|\left| #1 \right|\right| _{L_\infty}}
\newcommand{\normTwo}[1]{\left|\left| #1 \right|\right| _{2}}

\newcommand{\normLOne}[1]{\left|\left| #1 \right|\right| _{L_1}}
\newcommand{\normLTwo}[1]{\left|\left| #1 \right|\right| _{\mathcal{L}_2}}
\newcommand{\normInfInf}[1]{\left|\left| #1 \right|\right| _{\infty}}
\newcommand{\normTV}[1]{\left|\left| #1 \right|\right| _{\text{TV}}}

\newcommand{\signEq}[1]{\emph{\text{sign}}\left( #1 \right) }

\newcommand{\abs}[1]{\left| #1 \right|}
\newcommand{\keys}[1]{\left\{ #1 \right\}}
\newcommand{\sqbr}[1]{\left[ #1 \right]}
\newcommand{\brac}[1]{\left( #1 \right) }

\newcommand{\Real}[1]{\text{Re}\brac{ #1 }}
\newcommand{\Id}{\text{\em I}}
\renewcommand{\P}{\operatorname{\mathbb{P}}}
\newcommand{\MAT}[1]{\begin{bmatrix} #1 \end{bmatrix}}

\newcommand{\PROD}[2]{\left \langle #1, #2\right \rangle}
\newcommand{\SN}{S_{\text{near}}}
\newcommand{\SF}{S_{\text{far}}}
\newcommand{\SFEq}{S_{\emph{{\text{far}}}}}
\newcommand{\SNEq}{S_{\emph{\text{near}}}}
\newcommand{\Ker}{G}
\newcommand{\low}{\text{lo}}
\newcommand{\high}{\text{hi}}
\newcommand{\Klo}{K_{\low}}
\newcommand{\Klohat}{\widehat{K}_{\low}}
\newcommand{\Khi}{K_{\high}}
\newcommand{\KhiEq}{K_{\emph{\high}}}
\newcommand{\flo}{f_{\low}}
\newcommand{\floEq}{f_{\emph{\low}}}
\newcommand{\fhi}{f_{\high}}
\newcommand{\lambdalo}{\lambda_{\low}}
\newcommand{\lambdaloEq}{\lambda_{\emph{\low}}}
\newcommand{\lambdahi}{\lambda_{\high}}
\newcommand{\lambdahiEq}{\lambda_{\emph{\high}}}
\newcommand{\Qlo}{Q_{\low}}
\newcommand{\QloEq}{Q_{\emph{\low}}}
\newcommand{\Qhi}{Q_{\high}}
\newcommand{\Flo}{F_{\low}}
\newcommand{\FloEq}{F_{\emph{\low}}}
\newcommand{\dEq}{\emph{\text{d}}}
\newcommand{\xest}{x_{\text{est}}}
\newcommand{\xestEq}{x_{\emph{\text{est}}}}
\newcommand{\uest}{u_{\text{est}}}
\newcommand{\uestEq}{u_{\emph{\text{est}}}}
\newcommand{\normTVEq}[1]{\left|\left| #1 \right|\right| _{\emph{\text{TV}}}}

\author{Emmanuel J. Cand\`{e}s\thanks{Departments of Mathematics and
    of Statistics, Stanford University, Stanford CA} \,\, and Carlos
  Fernandez-Granda\thanks{Department of Electrical Engineering,
    Stanford University, Stanford CA}}
\title{Super-Resolution from Noisy Data}
\date{October 2012; Revised April 2013}

\begin{document}

\maketitle

\begin{abstract}
  This paper studies the recovery of a superposition of point sources
  from noisy bandlimited data.  In the fewest possible words, we only
  have information about the spectrum of an object in the
  low-frequency band $[-\flo, \flo]$ and seek to obtain a higher
  resolution estimate by extrapolating the spectrum up to a frequency
  $\fhi > \flo$. We show that as long as the sources are separated by
  $2/\flo$, solving a simple convex program produces a stable estimate
  in the sense that the approximation error between the
  higher-resolution reconstruction and the truth is proportional to
  the noise level times the square of the \emph{super-resolution
    factor} (SRF) $\fhi/\flo$.
\end{abstract}


\section{Introduction}

It is often of great interest to study the fine details of a signal at
a scale beyond the resolution provided by the available
measurements. In a general sense, super-resolution techniques seek to
recover high-resolution information from coarse-scale
data. There is a gigantic literature on this subject as
researchers try to find ways of breaking the
diffraction limit---a fundamental limit on the possible
resolution---imposed by most imaging systems.  Examples of
applications include conventional optical
imaging~\cite{superres_survey}, astronomy~\cite{astronomy_puschmann},
medical imaging~\cite{medical_greenspan}, and
microscopy~\cite{microscopy_mccutchen}. In electronic imaging, photon
shot noise limits the pixel size, making super-resolution techniques
necessary to recover sub-pixel
details~\cite{book_milanfar,imaging_park}. Among other fields
demanding and developing super-resolution techniques, one could cite
spectroscopy~\cite{spectroscopy_superresolution},
radar~\cite{radar_odendaal}, non-optical medical
imaging~\cite{pet_kennedy} and geophysics~\cite{seismic_khaidukov}.

In many of these applications, the signal we wish to super-resolve is
a superposition of point sources; depending upon the situation, these
may be celestial bodies in astronomy~\cite{pointsource_astronomy},
molecules in microscopy~\cite{pointsource_microscopy}, or line spectra in speech analysis~\cite{pointsource_speech}. A large part of the literature on super-resolution revolves around the problem of distinguishing two blurred point sources that are close together, but there has been much less analysis on the conditions under which it is possible to super-resolve the location of a large number of point sources with high precision. This question is of crucial importance, for instance, in fluorescence microscopy. Techniques such as photoactivated localization microscopy (PALM)~\cite{palm,fpalm} or stochastic optical reconstruction microscopy (STORM)~\cite{storm} are based on the use of probes that switch randomly between a fluorescent and a non-fluorescent state. To super-resolve a certain object, multiple frames are gathered and combined. Each frame consists of a superposition of blurred light sources that correspond to the active probes and are mostly well separated.

In the companion article~\cite{superres}, the authors studied the problem of recovering superpositions of point sources in a noiseless setting, where one has perfect low-frequency information. In contrast, the present paper considers a setting where the data are contaminated with noise, a situation which is unavoidable in practical
applications. In a nutshell, \cite{superres} proves that with noiseless
data, one can recover a superposition of point sources exactly,
namely, with arbitrary high accuracy, by solving a simple convex
program. This phenomenon holds as long as the spacing between the
sources is on the order of the resolution limit. With noisy data now,
it is of course no longer possible to achieve infinite precision. In
fact, suppose the noise level and sensing resolution are fixed. Then
one expects that it will become increasingly harder to recover the
fine details of the signal as the scale of these features become
finer. The goal of this paper is to make this vague statement
mathematically precise; we shall characterize the estimation error as
a function of the noise level and of the resolution we seek to
achieve, showing that it is in fact possible to super-resolve point sources from noisy data with high precision via convex optimization.

\subsection{The super-resolution problem}
\label{sec:sr_problem}
\begin{figure}
\centering
\includegraphics[width=13cm]{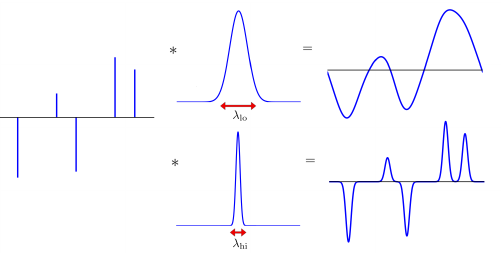}
\caption{Sketch of the super-resolution factor (SRF). A signal (left)
  is measured at a low resolution by a convolution with a kernel (top
  middle) of width $\lambdalo$ (top right). Super-resolution aims at
  approximating the outcome of a convolution with a much narrower
  kernel (bottom middle) of width $\lambdahi$. Hence, the goal is to
  recover the bottom right curve.}
\label{fig:srf}
\end{figure}
To formalize matters, we have observations about an object $x$ of the
form
\begin{equation}
  \label{eq:model}
  y(t) = (\Qlo x)(t) + z(t), 
\end{equation}
where $t$ is a continuous parameter (time, space, and so on) belonging
to the $d$-dimensional cube $[0,1]^d$. 
Above, $z$ is a noise term which can either be stochastic or deterministic,
and $\Qlo$ is a bandlimiting operator with a frequency cut-off equal
to $\flo=1/\lambdalo$.  Here, $\lambdalo$ is a positive parameter
representing the finest scale at which $x$ is observed.  To make this
more precise, we take $\Qlo$ to be a low-pass filter of width
$\lambdalo$ as illustrated at the top of Figure~\ref{fig:srf}; that
is,
\[
(\Qlo x)(t) = (\Klo * x)(t)
\]
such that in the frequency domain the convolution equation becomes 
\[
(\widehat{\Qlo x})(f) = \Klohat(f) \hat{x}(f), \quad f \in \Z^d.
\]
Here and henceforth we denote the usual
Fourier transform of a measure or function $g$, provided that it exists, by $\hat g(f) = \int e^{-i2\pi \<f, t\>} g(\text{d}t)$.  The spectrum of the low-pass kernel $\Klohat(f)$ vanishes
outside of the cell $[-\flo, \flo]^d$. 

Our goal is to resolve the signal $x$ at a finer scale $\lambdahi \ll
\lambdalo$. In other words, we would like to obtain a
\emph{high-resolution} estimate $\xest$ such that $\Qhi \, \xest
\approx \Qhi \, x$, where $\Qhi$ is a bandlimiting operator with
cut-off frequency $\fhi :=1/\lambdahi > \flo$. This is illustrated at
the bottom of Figure~\ref{fig:srf}, which shows the convolution
between $\Khi$ and $x$.  A different way to pose the problem is as
follows: we have noisy data about the spectrum of an object of
interest in the low-pass band $\sqbr{-\flo,\flo}$, and would like to
estimate the spectrum in the possibly much wider band
$\sqbr{-\fhi,\fhi}$. We introduce the super-resolution factor (SRF) as: 
\begin{equation}
\label{eq:srf}
\srf :=  \frac{\fhi}{\flo} = \frac{\lambdalo}{\lambdahi};  
\end{equation}
in words, we wish to double the resolution if the SRF is equal to two,
to quadruple it if the SRF equals four, and so on.  Given the
notorious ill-posedness of spectral extrapolation, a natural question
is how small the error at scale $\lambdahi$ between the estimated and
the true signal $\Khi \ast (\xest - x)$ can be? In
particular, how does it scale with both the noise level and the SRF?
This paper addresses this important question.

\subsection{Models and methods}

As mentioned earlier, we are interested in superpositions of point
sources modeled as
\begin{equation*}
  x = \sum_j a_j \delta_{t_j}, 
\end{equation*}
where $\{t_j\}$ are points from the interval $[0,1]$, $\delta_{\tau}$
is a Dirac measure located at $\tau$, and the amplitudes $a_j$ may be
complex valued. Although we focus on the one-dimensional case, our
methods extend in a straightforward manner to the multidimensional
case, as we shall make precise later on.  We assume the model
\eqref{eq:model} in which $t \in [0,1]$, which from now on we identify
with the unit circle $\mathbb{T}$, and $z(t)$ is a bandlimited error
term obeying
\begin{equation}
  \label{eq:error}
\|z\|_{L_1} = \int_{\mathbb{T}} |z(t)| \, \text{d}t \le \delta. 
\end{equation}
The measurement error $z$ is otherwise arbitrary and can be
adversarial. For concreteness, we set $\Klo$ to be the periodic
Dirichlet kernel
\begin{equation}
  \label{eq:dirichlet}
\Klo(t) = \sum_{k = -\flo}^{\flo} e^{i2\pi k t} = \frac{\sin(\pi(2\flo
  + 1)t)}{\sin(\pi t)}. 
\end{equation}
By definition, for each $f \in \Z$, this kernel obeys
$\Klohat(f) = 1$ if $|f| \le \flo$ whereas
$\Klohat(f)=0$ if $|f| > \flo$. We emphasize, however, that our
results hold for other low-pass filters. Indeed, our model
\eqref{eq:model} can be equivalently written in the frequency domain
as $\hat{y}(f) = \hat{x}(f) + \hat{z}(f)$, $|f| \le \flo$. Hence, if
the measurements are of the form $y=G_{\text{lo}}*x + z$ for some
other low-pass kernel $G_{\text{lo}}$, we can filter them linearly to obtain $\hat{y}_G(f):=\hat{y}(f)/\widehat{G}_{\text{lo}}(f) = \hat{x}(f) + \hat{z}(f)/\widehat{G}_{\text{lo}}(f)$. Our results can then be applied to this formulation if the weighted perturbation $\hat{z}(f)/\widehat{G}_{\text{lo}}(f)$ is bounded.

To perform recovery, we propose solving
\begin{equation}
\label{TVproblem_relaxed}
\min_{\tilde x} \,  \normTV{\tilde x} 
\quad \text{subject to} \quad \normLOne{\Qlo \tilde x - y} \leq \delta, 
\end{equation}
Above, $\normTV{x}$ is the total-variation norm of a measure (see
Chapter 6 of~\cite{rudin} or Appendix A in~\cite{superres}), which can
be interpreted as the generalization of the $\ell_1$ norm to the real
line. (If $x$ is a probability measure, then $\normTV{x} = 1$.) This
is not to be confused with the total variation of a function, a
popular regularizer in signal processing and computer vision. Finally, it is important to observe that the recovery algorithm is completely
agnostic to the target resolution~$\lambdahi$, so our results hold
simultaneously for any value of $\lambdahi >\lambdalo$. 

\subsection{Main result}

Our objective is to approximate the signal up until a certain
resolution determined by the width of the smoothing kernel $\lambdahi
> \lambdalo$ used to compute the error. To fix ideas, we set
\begin{align}
\label{eq:fej}
 \Khi(t)= \frac{1}{\fhi +1} \sum_{k = -\fhi}^{\fhi}\brac{\fhi
  + 1 - \abs{k}} e^{i2\pi k t} = \frac{1}{\fhi +1}\brac{ \frac{\sin(\pi(\fhi
  + 1)t)}{\sin(\pi t)}}^2
\end{align}
to be the Fej\'er kernel with cut-off frequency $\fhi =
1/\lambdahi$. Figure~\ref{fig:fejer} shows this kernel together with
its spectrum.
\begin{figure}
\centering
\subfloat[][]{
\includegraphics[width=5.5cm]{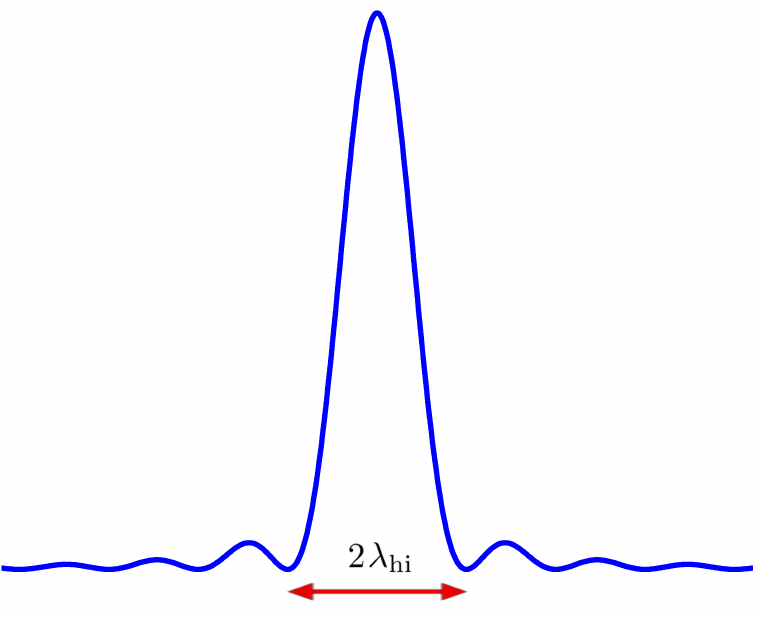}
}
\subfloat[][]{
\includegraphics[width=5.5cm]{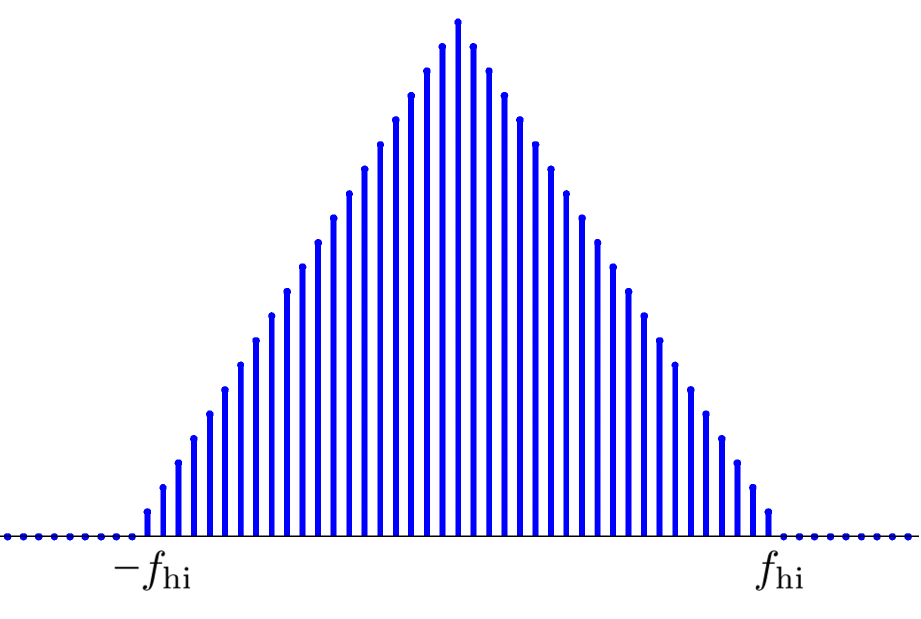}
}
\caption{The Fej\'er kernel \eqref{eq:fej} (a) with half width about
  $\lambdahi$, and its Fourier series coefficients (b). The kernel is
  bandlimited since the Fourier coefficients vanish beyond the cut-off
  frequency $\fhi$.}
\label{fig:fejer}
\end{figure}
As explained in Section~3.2 of~\cite{superres}, no matter what method
is used to achieve super-resolution, it is necessary to introduce a
condition about the support of the signal, which prevents the sources
from being too clustered together. Otherwise, the problem is easily shown to
be hopelessly ill-posed by leveraging Slepian's work on prolate
spheroidal sequences \cite{slepian_discrete}. In this paper, we use
the notion of minimum separation.
\begin{definition}[Minimum separation] For a family of points $T \subset
  \mathbb{T}$, the minimum separation is defined as the closest
  distance between any two elements from $T$,
  \begin{equation*}
    \Delta(T) = \inf_{(t, t') \in T \, : \, t \neq t'} \, \, |t - t'|. 
  \end{equation*}
\end{definition}

Our model \eqref{eq:error} asserts that we can achieve a low-resolution
error obeying
\[
\normLOne{\Klo*(\xest- x)} \le \delta,
\]
but that we cannot do better as well. The main question is: how does
this degrade when we substitute the low-resolution with the
high-resolution kernel? 
\begin{theorem}
\label{theorem:noise}
Assume that the support $T$ of $x$ obeys the separation condition
\begin{align}
\Delta(T) \geq 2 \lambdaloEq. \label{eq:min-distance}
\end{align}
Then under the noise model~\eqref{eq:error}, any solution $\xestEq$ to
problem~\eqref{TVproblem_relaxed}\footnote{To be precise, the theorem holds for any feasible point $\tilde{x}$ obeying $\normTV{\tilde{x}}\leq \normTV{x}$; this set is not empty since it contains $x$.} obeys
\begin{equation*}
  \normLOne{\KhiEq*(\xestEq- x)} \leq C_0 \, \srfeq^2 \, \delta,  
\end{equation*}
where $C_0$ is a positive numerical constant. 
\end{theorem}
Thus, minimizing the total-variation norm subject to data constraints
yields a stable approximation of any superposition of Dirac measures
obeying the minimum-separation condition.  When $z=0$, setting
$\delta=0$ and letting \srf$\rightarrow \infty$, this recovers the
result in~\cite{superres} which shows that $\xest = x$, i.e.~we
achieve infinite precision. What is interesting here is the quadratic
dependence of the estimation error in the super-resolution factor.

We have chosen to analyze problem~\eqref{TVproblem_relaxed} and a perturbation with bounded $L_1$ norm for simplicity, but our techniques can be adapted to other recovery schemes and noise models. For instance, suppose we observe noisy
samples of the spectrum
\begin{equation}
  \label{eq:stochasticnoise}
\eta(k) = \int_{\mathbb{T}} e^{-i2\pi k t} \, x(d t) + \epsilon_k, \quad
k = -\flo, -\flo + 1, \ldots, \flo,
\end{equation}
where $\epsilon_k$ is an iid sequence of complex-valued
$\mathcal{N}(0, \sigma^2)$ variables (this means that the real and
imaginary parts are independent $\mathcal{N}(0, \sigma^2)$ variables).
This is equivalent to a line-spectra estimation problem with additive Gaussian white noise, as we explain below. In order to super-resolve the signal under this model, we propose the following convex program
\begin{equation}
\label{TVproblem_relaxed_L2}
\min_{\tilde x} \,  \normTV{\tilde x} 
\quad \text{subject to} \quad \normLTwo{\Qlo \tilde x - y} \leq \delta, 
\end{equation}
which can be implemented using off-the-shelf software as discussed in Section~\ref{sec:numerical}. A corollary to our main theorem establishes that with high probability solving this problem allows to super-resolve the signal despite the added perturbation with an error that scales with the square of the super-resolution factor and is proportional to the noise level.
\begin{corollary}
\label{cor:stochastic}
Fix $\gamma > 0$. Under the stochastic noise model
\eqref{eq:stochasticnoise}, the solution to problem~\eqref{TVproblem_relaxed_L2} with $\delta = \brac{1+\gamma}\sigma \sqrt{ 4 \floEq +2}$ obeys
\begin{equation}
\label{eq:stochastic}
\normLOne{\KhiEq*(\xestEq- x)} \leq C_0\, \brac{1+\gamma}\, \sqrt{4 \floEq +2} \, \srfeq^2 \, \sigma
\end{equation}
with probability at least $1-e^{-2 \floEq \gamma^2}$. 
\end{corollary}
This result is proved in Section~\ref{sec:stochastic} of the appendix.

\subsection{Extensions}

{\em Other high-resolution kernels.} We work with the high-resolution
Fej\'er kernel but our results hold for any symmetric kernel that obeys the properties~\eqref{eq:kernel_cond12} and~\eqref{eq:kernel_cond3} below,
since our proof only uses these simple estimates. The first reads
\begin{equation} 
  \int_{\mathbb{T}} \abs{\Khi\brac{t}} \text{d}t \leq C_0,
  \qquad 
  \int_{\mathbb{T}} \abs{\Khi'\brac{t}} \text{d}t \leq
  C_1 \, {\lambdahi^{-1}}, \qquad \sup
  \abs{\Khi''\brac{t}} \leq
  C_2 \, \lambdahi^{-3}, \label{eq:kernel_cond12}
\end{equation}
where $C_0$, $C_1$ and $C_2$ are positive constants independent of
$\lambdahi$. The second is that there exists a nonnegative and
nonincreasing function $f : [0,1/2] \rightarrow \R$ such that 
\[
\abs{\Khi''\brac{t+\lambdahi}} \le f(t), \quad 0 \le t \le 1/2, 
\]
and
\begin{equation}
  \int_{0}^{1/2} f(t) \text{d}t \leq
  C_3 \, \lambdahi^{-2}. \label{eq:kernel_cond3}
\end{equation}
This is to make sure that \eqref{eq:gfp2} holds.  (For the Fej\'er
kernel, we can take $f$ to be quadratic in $\sqbr{0,1/2-\lambdahi}$
and constant in $\sqbr{1/2-\lambdahi,1/2}$.)

{\em Higher dimensions.} Our techniques can be applied to establish
robustness guarantees for the recovery of point sources in higher
dimensions.  The only parts of the proof of
Theorem~\ref{theorem:noise} that do not generalize directly are
Lemmas~\ref{lemma:dualpol},~\ref{lemma:dualpol_extra}
and~\ref{lemma:int_t_dh_polynomial}. However, the methods used to
prove these lemmas can be extended without much difficulty to multiple
dimensions as described in Section~\ref{sec:multidim} of the Appendix.

{\em Spectral line estimation.} Swapping time and frequency,
Theorem~\ref{theorem:noise} can be immediately applied to the
estimation of spectral lines in which we observe
\[
y(t) = \sum_j \alpha_j e^{i2\pi \omega_j t} + z(t), \quad t = 0, 1,
\ldots, n-1,  
\]
where $\alpha$ is a vector of complex-valued amplitudes and $z$ is a noise term. Here, our work implies that a non-parametric
method based on convex optimization is capable of approximating the
spectrum of a multitone signal with arbitrary frequencies, as long as
these frequencies are sufficiently far apart, and furthermore that the
reconstruction is stable. In this setting, the smoothed error quantifies the quality of the approximation windowed at a certain spectral
resolution.

\subsection{Related work}

Earlier work on the super-resolution problem in the presence of noise studied under which conditions recovery is not hopelessly ill-posed, establishing that sparsity is not sufficient even for signals supported on a grid~\cite{donohoSuperres,slepian_discrete}. More recently, \cite{batenkov_prony} studies the local stability of the problem in a continuous domain. These works, however, do not provide any tractable algorithms to perform recovery. 

Since at least the work of Prony~\cite{prony}, parametric methods based on polynomial rooting have been a popular approach to the super-resolution of trains of spikes and, equivalently, of line spectra. These techniques are typically based on the eigendecomposition of a sample covariance matrix of the data~\cite{music1,music2,esprit,blu_fri}. The theoretical analysis available for these methods is based on an asymptotic characterization of the sample covariance matrices under Gaussian noise~\cite{stoica_statistical,clergeot_music}, which unfortunately does not allow to obtain explicit guarantees on the recovery error beyond very simple cases involving one or two spikes. Other works extend these results to explore the trade-off between resolution and signal-to-noise ratio for the detection of two closely-spaced line spectra~\cite{milanfar_sinusoids} or light sources~\cite{helstrom,milanfar_imaging}. A recent reference~\cite{fannjiang_music}, which focuses mainly on the related problem of imaging point scatterers, analyzes the performance of a parametric method in the case of signals sampled randomly from a discrete grid under the assumption that the sample covariance matrix is close enough to the true one. 
In general, parametric techniques require prior knowledge of the model order and rely heavily on the assumption that the noise is white or at least has known spectrum (see Chapter 4 of~\cite{stoica_book}). An alternative approach that overcomes the latter drawback is to perform nonlinear least squares estimation of the model parameters~\cite{stoicaNLS}. Unfortunately, the resulting optimization problem has an extremely multimodal cost function, which makes it very sensitive to initialization~\cite{stoicaNLSmultimodal}. 

The total-variation norm is the continuous analog of the $\ell_1$ norm
for finite dimensional vectors, so our recovery algorithm can be
interpreted as finding the shortest linear combination---in an
$\ell_1$ sense---of elements taken from a continuous and infinite
dictionary. Previous theoretical work on the stability of this approach is limited to a discrete and finite-dimensional setting, where the support of the signal of interest is restricted to a finer uniform grid~\cite{superres}. Even if we discretize the dictionary, other stability
results for sparse recovery in redundant dictionaries do not apply due to the high
coherence between the elements. In addition, it is worth mentioning that working with a discrete
dictionary can easily degrade the quality of the
estimate~\cite{mismatch} (see~\cite{stoica_grid} for a related discussion concerning grid selection for spectral analysis), which highlights the importance of analyzing the problem in the continuous domain. This observation has spurred the appearance
of modified compressed-sensing techniques specifically tailored to the
task of spectral
estimation~\cite{spectral_cs,refinement,grids_fannjiang}. Proving
stability guarantees for these methods under conditions on the support
or the dynamic range of the signal is an interesting research
direction.

\section{Proof of Theorem \ref{theorem:noise}}
\label{sec:noise_proof}

It is useful to first introduce various objects we shall need in the
course of the proof. We let $T = \{t_j\}$ be the support of $x$ and
define the disjoint subsets
\begin{align*}
  \SN^{\lambda} \brac{j} & := \keys{ t  \; : \; \abs{t-t_j}\leq 0.16 \lambda},\\
  \SF^{\lambda} &:= \keys{ t \; : \; \abs{t-t_j}> 0.16 \lambda, \;
    \forall t_j \in T};
\end{align*}
here, $\lambda \in \keys{\lambdalo,\lambdahi}$, and $j$ ranges from 1
to $\abs{T}$. We write the union of the sets $\SN^{\lambda}\brac{j}$
as
\begin{align*}
  \SN^{\lambda} &:= \cup _{j=1}^{\abs{T}} \SN^{\lambda}\brac{j}
 \end{align*}
 and observe that the pair $(\SN^{\lambda}, \SF^{\lambda})$ forms a
 partition of $\mathbb{T}$.  The value of the constant $0.16$ is not
 important and chosen merely to simplify the argument. We denote the
 restriction of a measure $\mu$ with finite total variation on a set
 $S$ by $P_{S} \mu$ (note that in contrast we denote the low-pass
 projection in the frequency domain by $\Qlo$). This restriction is
 well defined for the above sets, as one can take the Lebesgue
 decomposition of $\mu$ with respect to a positive $\sigma$-finite
 measure supported on any of them \cite{rudin}. To keep some expressions in compact form, we set
 \begin{align*}
  I_{\SN^{\lambda} \brac{j}}\brac{\mu} &:=
  \frac{1}{\lambdalo^2}\int_{\SN^{\lambda} \brac{j}} \brac{t-t_j}^2
  \abs{\mu} \brac{\text{d} t},\\ I_{\SN^{\lambda}}\brac{\mu} &:= \sum_{t_j \in
    T} I_{\SN^{\lambda} \brac{j}}\brac{\mu}
 \end{align*}
 for any measure $\mu$ and $\lambda \in
 \keys{\lambdalo,\lambdahi}$. Finally, we reserve the symbol $C$ to
 denote a numerical constant whose value may change at each
 occurrence.

Set $h= x- \xest$. The error obeys 
\begin{align*}
  \normLOne{\Qlo h } \leq \normLOne{\Qlo x -y}+\normLOne{y-\Qlo
    \xest}\leq 2\delta,
\end{align*}
and has bounded total-variation norm since $\normTV{h}\leq
\normTV{x}+\normTV{\xest} \leq 2\normTV{x}$. Our aim is to bound the
$L_1$ norm of the smoothed error $e := \Khi * h$, 
\begin{align*}
 \normLOne{ e} & = \int_{\mathbb{T}} \abs{\int_{\mathbb{T}}\Khi\brac{t-\tau}  h\brac{\text{d}\tau} }\text{d}t.
\end{align*}

We begin with a lemma bounding the total-variation norm of $h$ `away'
from $T$. 
\begin{lemma}
\label{lemma:TC_bound}
Under the conditions
of Theorem \ref{theorem:noise}, there exist positive constants $C_a$
and $C_b$ such that
\begin{align*}
\normTVEq{P_{\SFEq^{\lambdaloEq}}\brac{h}} + I_{\SNEq^{\lambdaloEq}}\brac{h}   & \leq  C_a \, \delta,\\
\normTVEq{P_{\SFEq^{\lambdahiEq}}\brac{h}} & \leq  C_b \, \srfeq^2 \, \delta.
\end{align*}
\end{lemma}
This lemma is proved in Section \ref{sec:TC_bound} and relies on the
existence of a low-frequency dual polynomial constructed in
\cite{superres} to guarantee exact recovery in the noiseless setting.

To develop a bound about $\|e\|_{L_1}$,
we begin by applying the triangle inequality to obtain
\begin{align}
\abs{e\brac{t}} = \abs{\int_{\mathbb{T}} \Khi\brac{t-\tau}  h\brac{\text{d}\tau}} & \leq \abs{\int_{\SF^{\lambdahi}}\Khi\brac{t-\tau}  h\brac{\text{d}\tau}} + \abs{\int_{\SN^{\lambdahi}} \Khi\brac{t-\tau} h\brac{\text{d}\tau}}. \label{split_integrand}
\end{align} 
By a corollary of the Radon-Nykodim Theorem (see Theorem 6.12 in
\cite{rudin}), it is possible to perform the polar decomposition $
P_{\SF^{\lambdahi}} \brac{h}\brac{\text{d}\tau} = e^{i 2 \pi \theta
  \brac{\tau}}\abs{P_{\SF^{\lambdahi}} \brac{h}}\brac{\text{d}\tau} $
such that $\theta\brac{\tau}$ is a real function and
$\abs{P_{\SF^{\lambdahi}} \brac{h}}$ is a positive measure. Then
\begin{align}
\int_{\mathbb{T}} \abs{\int_{\SF^{\lambdahi}}\Khi\brac{t-\tau}  h\brac{\text{d}\tau}}\text{d}t & \leq \int_{\mathbb{T}} \int_{\SF^{\lambdahi}} \abs{\Khi\brac{t-\tau}}  \abs{P_{\SF^{\lambdahi}} \brac{h}}\brac{\text{d}\tau}\text{d}t \notag\\
&=  \int_{\SF^{\lambdahi}}\brac{\int_{\mathbb{T}} \abs{\Khi\brac{t-\tau}}\text{d}t}  \abs{P_{\SF^{\lambdahi}} \brac{h}}\brac{\text{d}\tau} \notag\\
& \leq  C_0 \normTV{ P_{\SF^{\lambdahi}} \brac{h}}, \label{boundA}
\end{align}
where we have applied Fubini's theorem and \eqref{eq:kernel_cond12}
(note that the total-variation norm of $\abs{P_{\SF^{\lambdahi}}
  \brac{h}}$ is bounded by $2\normTV{x}<\infty$).

In order to control the second term in the right-hand side of
\eqref{split_integrand}, we use a first-order approximation of the
super-resolution kernel provided by the Taylor series expansion of
$\psi\brac{\tau}=\Khi\brac{t-\tau}$ around $t_j$: for any $\tau$ such
that $\abs{\tau-t_j}\leq 0.16 \lambdahi$, we have
\begin{align*}
  \abs{ \Khi\brac{t-\tau} - \Khi\brac{t-t_j} -
    \Khi'\brac{t-t_j}\brac{t_j - \tau}} & \leq \sup_{u : \abs{t-t_j-
      u} \leq 0.16 \lambdahi} \, \frac12 {\abs{\Khi''(u)}
    \brac{\tau-t_j}^2}.
\end{align*}
Applying this together with the triangle inequality, and setting
$t_j=0$ without loss of generality, give
\begin{multline}
  \int_{\mathbb{T}} \abs{\int_{\SN^{\lambdahi}\brac{j} } \Khi\brac{t-\tau} h\brac{\text{d}\tau}} \text{d}t  \leq  \int_{\mathbb{T}} \abs{\int_{\SN^{\lambdahi}\brac{j} } \Khi\brac{t} h\brac{\text{d}\tau}} \text{d}t \\
   + \int_{\mathbb{T}} \abs{\int_{\SN^{\lambdahi}\brac{j} } \Khi'\brac{t}\tau h\brac{\text{d}\tau}} \text{d}t 
   + \frac{1}{2}\int_{\mathbb{T}} \abs{\int_{\SN^{\lambdahi}\brac{j}}
    \sup_{\abs{t-u} \leq 0.16 \lambdahi} \abs{\Khi''\brac{u}} \tau^2
    |h|\brac{\text{d}\tau}} \text{d}t.  \label{boundB}
\end{multline}
(To be clear, we do not lose generality by setting $t_j=0$ since the
analysis is invariant by translation; in particular by a translation
placing $t_j$ at the origin. To keep things as simple as possible, we
shall make a frequent use of this argument.) We then combine Fubini's
theorem with \eqref{eq:kernel_cond12} to obtain
\begin{equation}
  \int_{\mathbb{T}} \abs{\int_{\SN^{\lambdahi}\brac{j} } \Khi\brac{t} h\brac{\text{d}\tau}} \text{d}t \leq  \int_{\mathbb{T}} \abs{\Khi\brac{t} }\text{d}t \abs{\int_{\SN^{\lambdahi}\brac{j} }  h\brac{\text{d}\tau}} 
  \leq C_0 \abs{\int_{\SN^{\lambdahi}\brac{j} }  h\brac{\text{d}\tau}}\label{boundErr0}
\end{equation}
and
\begin{equation}
  \int_{\mathbb{T}} \abs{\int_{\SN^{\lambdahi}\brac{j} } \Khi'\brac{t}\tau h\brac{\text{d}\tau}} \text{d}t \leq  \int_{\mathbb{T}} \abs{\Khi'\brac{t} }\text{d}t \abs{\int_{\SN^{\lambdahi}\brac{j} } \tau h\brac{\text{d}\tau}}
  \leq \frac{C_1}{\lambdahi} \abs{\int_{\SN^{\lambdahi}\brac{j} }  \tau h\brac{\text{d}\tau}}\label{boundErr1}.
\end{equation}
Some simple calculations show that \eqref{eq:kernel_cond12} and
\eqref{eq:kernel_cond3} imply
\begin{align}
\int_{\mathbb{T}} \sup_{\abs{t-u} \leq 0.16 \lambdahi} \abs{\Khi''\brac{u}} \text{d}t & \leq \frac{C_4}{ \lambdahi^2} \label{eq:gfp2}
\end{align}
for a positive constant $C_4$. This together with Fubini's theorem
yield
\begin{align}
  \int_{\mathbb{T}} \abs{\int_{\SN^{\lambdahi}\brac{j}} \abs{\Khi''\brac{u}} \tau^2 |h|\brac{\text{d}\tau}} \text{d}t & \leq  \int_{\mathbb{T}} \sup_{\abs{t-u} \leq 0.16 \lambdahi} \abs{\Khi''\brac{t} }\text{d}t \abs{\int_{\SN^{\lambdahi}\brac{j} } \tau^2 |h|\brac{\text{d}\tau}} \notag \\
  & \leq C_4 \, \srf^2 \,
  I_{\SN^{\lambdahi}\brac{j}}\brac{h}\label{boundErr2}
\end{align}
for any $u$. 
In order to make use of these bounds, it is
necessary to control the local action of the measure $h$ on a constant
and a linear function. The following two lemmas are proved in
Sections~\ref{sec:int_dh} and~\ref{sec:int_t_dh}.
\begin{lemma}
\label{lemma:int_dh}
Take $T$ as in Theorem~\ref{theorem:noise} and any measure $h$ obeying
$\normLOne{\QloEq h} \leq 2 \delta$. Then
\[
\sum_{t_j \in T} \abs{ \int_{\SNEq^{\lambdahiEq}\brac{j} }
  h\brac{ \dEq \tau}} \leq 2\delta + \normTVEq{P_{\SFEq^{\lambdahiEq}}
  \brac{h}} + C \, I_{\SNEq^{\lambdahiEq}}\brac{h}.
  \]
\end{lemma}
\begin{lemma}
\label{lemma:int_t_dh}
Take $T$ as in Theorem~\ref{theorem:noise} and any measure $h$ obeying
$\normLOne{\QloEq h} \leq 2 \delta$. Then 
 \[
 \sum_{t_j \in T} \abs{\int_{\SNEq^{\lambdahiEq}\brac{j} } \brac{\tau-t_j}
   h\brac{ \dEq \tau}} \leq C \brac{ \lambdaloEq \, \delta +
   \lambdaloEq \, \normTVEq{P_{\SFEq^{\lambdaloEq}} \brac{h}} + \lambdaloEq \,
   I_{\SNEq^{\lambdaloEq} }\brac{h} + \lambdahiEq \, \srfeq^2 \,
   I_{\SNEq^{\lambdaloEq} }\brac{h} }.
 \]
\end{lemma}

We may now conclude the proof of our main theorem.  Indeed, the
inequalities \eqref{boundA}, \eqref{boundB}, \eqref{boundErr0},
\eqref{boundErr1} and \eqref{boundErr2} together with $
I_{\SN^{\lambdahi} }\brac{h} \leq I_{\SN^{\lambdalo} }\brac{h}$ imply
\[
\normLOne{e} \leq C \brac{ \srf \, \delta +
  \normTV{P_{\SF^{\lambdahi}} \brac{h}} + \srf \,
  \normTV{P_{\SF^{\lambdalo}} \brac{h}} + \srf^2 \, I_{\SN^{\lambdalo}
  }\brac{h} } \leq C\;\srf^2 \delta,
\]
where the second inequality follows from Lemma~\ref{lemma:TC_bound}.

\subsection{Proof of Lemma \ref{lemma:TC_bound}}
\label{sec:TC_bound}

The proof relies on the existence of a certain low-frequency
polynomial, characterized in the following lemma which recalls results from Proposition 2.1 and Lemma 2.5 in
\cite{superres}.
\begin{lemma}
\label{lemma:dualpol}
Suppose $T$ obeys the separation condition \eqref{eq:min-distance} and
take any $v \in \C^{|T|}$ with $|v_j| = 1$. Then there exists a
low-frequency trigonometric polynomial
\begin{equation*}
q(t) = \sum_{k = -\floEq}^{\floEq} c_k e^{i2\pi k t} 
\end{equation*}
obeying the following properties: 
\begin{align}
q(t_j) & = v_j, \quad t_j \in T, \label{q0_1}\\
|q(t)| & \leq 1-\frac{C_a  \brac{t-t_j}^2}{\lambdaloEq^2} ,  \quad  t \in \SNEq^{\lambdaloEq}\brac{j}, \label{q0_2}\\
|q(t)| & < 1-C_b ,  \quad  t \in \SFEq^{\lambdaloEq}, \label{q0_3}
\end{align}
with $0<C_b \leq 0.16^2 C_a<1$.
\end{lemma}

Invoking a corollary of the Radon-Nykodim Theorem (see Theorem 6.12 in \cite{rudin}), it is possible to perform a polar decomposition of $\PTn h$,
\begin{align*}
 \PTn h = e^{i \phi \brac{t}}\abs{\PTn h}\text{,}
\end{align*}
such that $\phi\brac{t}$ is a real function defined on $\mathbb{T}$. To prove Lemma \ref{lemma:TC_bound}, we work with $v_j = e^{-i \phi(t_j)}$. Since
$q$ is low frequency,
\begin{align}
\abs{\int_{\mathbb{T}} q(t) \text{d}h (t)} & = \abs{\int_{\mathbb{T}} q(t) \Qlo h\brac{t} \text{d}t} \leq \normInf{q} \normLOne{\Qlo h} \leq 2 \delta. \label{eq:int_q}
\end{align}
Next, since $q$ interpolates $e^{-i \phi \brac{t}}$ on $T$,
\begin{align}
  \normTV{ \PTn h } = \int_{\mathbb{T}} q(t) \PTn h\brac{\text{d}t} & \leq \abs{\int_{\mathbb{T}} q(t) h\brac{\text{d}t}} + \abs{\int_{T^c}  q(t) h\brac{\text{d}t}} \notag \\
  & \leq 2\delta + \sum_{j\in T}
  \abs{\int_{\SN^{\lambdalo}\brac{j}\setminus\keys{t_j} } q(t)
    h\brac{\text{d}t}}+ \abs{\int_{\SF^{\lambdalo}} q(t)
    h\brac{\text{d}t}}. \label{PTh_1c}
\end{align}
Applying \eqref{q0_3} in Lemma \ref{lemma:dualpol} and H\"older's
inequality, we obtain
\begin{align}
\abs{\int_{\SF^{\lambdalo}}  q(t) h\brac{\text{d}t}} & \leq \normInf{P_{\SF^{\lambdalo}}q} \normTV{P_{\SF^{\lambdalo}}\brac{h}} \notag\\
& \leq \brac{1-C_b}\normTV{P_{\SF^{\lambdalo}}\brac{h}}.  \label{PTh_2c}
\end{align}
Set $t_j=0$ without loss of generality. The triangle inequality and
\eqref{q0_2} in Lemma~\ref{lemma:dualpol} yield
\begin{align}
\abs{\int_{\SN^{\lambdalo}\brac{j}\setminus\keys{0}}  q(t) h\brac{\text{d}t}} & \leq  \int_{\SN^{\lambdalo}\brac{j}\setminus\keys{0}}  \abs{q(t)} \abs{h}\brac{\text{d}t} \notag \\
& \leq  \int_{\SN^{\lambdalo}\brac{j}\setminus\keys{0}}  \brac{1-\frac{C_a  t^2}{\lambdalo^2}} \abs{h}\brac{\text{d}t} \notag \\
& \leq \int_{\SN^{\lambdalo}\brac{j}\setminus\keys{0}} \abs{h}\brac{\text{d}t} - C_a I_{\SN^{\lambdalo} \brac{j}}\brac{h} . \label{PTh_3c}
\end{align}
Combining \eqref{PTh_1c}, \eqref{PTh_2c} and \eqref{PTh_3c} gives 
\begin{align*}
\normTV{ \PTn h } \leq &  2\delta + \normTV{ \PTcn h } - C_b  \normTV{P_{\SF^{\lambdalo}}\brac{h}}-  C_a I_{\SN^{\lambdalo}}\brac{h}.
\end{align*}

Observe that we can substitute $\lambdalo$
with $\lambdahi$ in \eqref{PTh_1c} and \eqref{PTh_3c} and obtain 
\[
\normTV{ \PTn h } \leq   2\delta + \normTV{ \PTcn h } - 0.16^2 \, C_a \, \srf^{-2}  \normTV{P_{\SF^{\lambdahi}}\brac{h}}-  C_a I_{\SN^{\lambdahi}}\brac{h}.
\]
This follows from using \eqref{q0_2} instead of \eqref{q0_3} to bound the magnitude of $q$ on $\SF^{\lambdahi}$.  

These inequalities can be interpreted as a generalization of the
strong null-space property used to obtain stability guarantees for
super-resolution on a discrete grid (see Lemma 3.1 in
\cite{superres}).  Combined with the fact that $\hat{x}$ has minimal
total-variation norm among all feasible points, they yield
\begin{align*}
\normTV{ x } & \geq  \normTV{ x + h }\\
& \geq  \normTV{ x } - \normTV{ \PTn h } +  \normTV{\PTcn h } \\
& \geq  \normTV{ x } -2\delta + C_b  \normTV{P_{\SF^{\lambdalo}}\brac{h}} +  C_a I_{\SN^{\lambdalo}}\brac{h}.
\end{align*}
As a result, we conclude that
\begin{align*}
C_b \normTV{P_{\SF^{\lambdalo}}\brac{h}} + C_a I_{\SN^{\lambdalo}}\brac{h} \leq  2 \delta, 
\end{align*}
and by the same argument,
\begin{align*}
 0.16^2 \, C_a \, \srf^{-2} \normTV{P_{\SF^{\lambdahi}}\brac{h}} + C_a I_{\SN^{\lambdahi}}\brac{h} \leq  2 \delta.
\end{align*}
This finishes the proof. 

\subsection{Proof of Lemma \ref{lemma:int_dh}}
\label{sec:int_dh}

The proof of this lemma relies upon the low-frequency polynomial from Lemma~\ref{lemma:dualpol} and the fact that $q(t)$ is close to the chosen sign pattern when $t$ is near any element of the support. The following intermediate result is proved in
Section~\ref{sec:dualpol_extra} of the Appendix.
\begin{lemma}
\label{lemma:dualpol_extra}
There is a polynomial $q$ satisfying the properties from
Lemma~\ref{lemma:dualpol} and, additionally,
\begin{align*}
  |q(t)-v_j| & \leq \frac{C \brac{t-t_j}^2}{\lambdaloEq^2} , \quad
  \text{ for all } t \in \SNEq^{\lambdaloEq}\brac{j}.
 \end{align*}
\end{lemma}
Consider the polar form 
\begin{align*}
\int_{\SN^{\lambdahi}\brac{j}}  h\brac{\text{d} \tau} = \abs{\int_{\SN^{\lambdahi}\brac{j}}  h\brac{\text{d} \tau}} e^{i\theta_j},
\end{align*}
where $\theta_j \in [0,2\pi)$. We set $v_j=e^{i \theta_j}$ in Lemma~\ref{lemma:dualpol} and apply the triangular inequality to obtain
\begin{align}
 \abs{\int_{\SN^{\lambdahi}\brac{j} }  h\brac{\text{d} \tau}} & =   \int_{\SN^{\lambdahi}\brac{j} } e^{-i\theta_j} h\brac{\text{d} \tau} \notag\\
& \leq \int_{\SN^{\lambdahi}\brac{j} } q\brac{\tau} h\brac{\text{d} \tau}+\abs{ \int_{\SN^{\lambdahi}\brac{j} } \brac{ q\brac{\tau}-e^{-i\theta_j} }h\brac{\text{d} \tau}},\label{dh_bound}
\end{align}
for all $t_j \in T$. By another application of the triangle inequality and \eqref{eq:int_q}
\begin{align}
\int_{\SN^{\lambdahi}} q\brac{\tau} h\brac{\text{d} \tau} & \leq \abs{\int_{\mathbb{T}} q\brac{\tau} h\brac{\text{d} \tau}} + \abs{\int_{\SF^{\lambdahi}} q\brac{\tau} h\brac{\text{d} \tau}} \leq 2 \delta + \normTV{P_{\SF^{\lambdahi}} \brac{h}}. \label{ineq_delta_A}
\end{align}
To bound the remaining term in \eqref{dh_bound}, we apply Lemma~\ref{lemma:dualpol_extra} 
with $t_j=0$ (this is no loss of generality),
\begin{align*}
  \abs{\int_{\SN^{\lambdahi}\brac{j}  } \brac{q(t)-e^{-i\theta_j} } h\brac{\text{d} t}} & \leq \int_{\SN^{\lambdahi}\brac{j}  } \abs{q(t)-e^{-i\theta_j} } \abs{h} \brac{\text{d}t}\\
  & \leq \int_{\SN^{\lambdahi}\brac{j} } \frac{C t^2}{\lambdalo^2}
  \abs{h} \brac{\text{d}t} = C I_{\SN^{\lambdahi}\brac{j}}\brac{h}.
\end{align*}
It follows from this, \eqref{dh_bound} and \eqref{ineq_delta_A} that 
\begin{align*}
\abs{ \int_{\SN^{\lambdahi}} h\brac{\text{d} \tau}} & \leq 2\delta + \normTV{P_{\SF^{\lambdahi}} \brac{h}} + C I_{\SN^{\lambdahi}}\brac{h}.
\end{align*}
The proof is complete.

\subsection{Proof of Lemma \ref{lemma:int_t_dh}}
\label{sec:int_t_dh}

We record a simple lemma. 
\begin{lemma}
\label{lemma:int_t_dh_f_c}
For any measure $\mu$ and $t_j=0$,
\begin{align*}
  \abs{\int_{0.16 \lambdahiEq}^{0.16 \lambdaloEq} \tau \mu\brac{\dEq
      \tau}} & \leq 6.25 \, \lambdahiEq \, \srfeq^2 \,
  I_{\SNEq^{\lambdaloEq} \brac{j}}\brac{\mu}.
\end{align*}
\end{lemma}
\begin{proof}
  Note that in the interval $\sqbr{0.16 \lambdahi,0.16 \lambdalo}$,
  $t/0.16\,\lambdahi \ge 1$, whence
\[
\abs{\int_{0.16 \lambdahi}^{0.16 \lambdalo} \tau \mu \brac{\text{d}
    \tau}} \leq \int_{0.16 \lambdahi}^{0.16 \lambdalo} \tau \abs{\mu}
\brac{\text{d} \tau} \leq \int_{0.16 \lambdahi}^{0.16 \lambdalo}
\frac{\tau^2}{0.16 \lambdahi} \abs{\mu} \brac{\text{d} \tau} \leq
\frac{ \lambdalo^2 }{0.16 \, \lambdahi} I_{\SN^{\lambdalo}
  \brac{j}}\brac{\mu}.
\]
\end{proof}

We now turn our attention to the proof of Lemma \ref{lemma:int_t_dh}. 
By the triangle inequality,
\begin{multline}
  \sum_{t_j \in T} \abs{\int_{\SN^{\lambdahi}\brac{j} }
    \brac{\tau-t_j} h\brac{\text{d} \tau}} \leq \\
\sum_{t_j \in T}
  \abs{\int_{\SN^{\lambdalo}\brac{j} } \brac{\tau-t_j} h\brac{\text{d}
      \tau}} + \sum_{t_j \in T} \abs{\int_{0.16 \lambdahi \leq
      \abs{\tau-t_j} \leq 0.16 \lambdalo} \brac{\tau-t_j}
    h\brac{\text{d} \tau}} . \label{t_dh_bound_1}
\end{multline}
The second term is bounded via Lemma \ref{lemma:int_t_dh_f_c}. For the
first, we use an argument very similar to the proof of
Lemma~\ref{lemma:int_dh}. Here, we exploit the existence of a
low-frequency polynomial that is almost linear in the vicinity of the
elements of $T$. The result below is proved in
Section~\ref{sec:int_t_dh_polynomial} of the Appendix.
\begin{lemma}
\label{lemma:int_t_dh_polynomial}
Suppose $T$ obeys the separation condition \eqref{eq:min-distance} and
take any $v \in \C^{|T|}$ with $|v_j| = 1$. Then there exists a
low-frequency trigonometric polynomial
\begin{equation*}
\label{eq:trig_1}
q_1(t) = \sum_{k = -\floEq}^{\floEq} c_k e^{i2\pi k t} 
\end{equation*}
obeying 
\begin{align}
|q_1(t)-v_j\brac{t-t_j}| & \leq \frac{C_a \brac{t-t_j}^2}{\lambdaloEq } ,  \quad  t \in \SNEq^{\lambdaloEq}\brac{j}, \label{q1_1}\\
|q_1(t)| & \leq C_b \lambdaloEq ,  \quad  t \in \SFEq^{\lambdaloEq}, \label{q1_2}
\end{align}
for positive constants $C_a$, $C_b$. 
\end{lemma}
Consider the polar decomposition of 
\begin{align*}
\int_{\SN^{\lambdalo}\brac{j} }  \brac{\tau-t_j} h\brac{\text{d} \tau} = \abs{\int_{\SN^{\lambdalo}\brac{j} } \brac{\tau-t_j} h\brac{\text{d} \tau}} e^{i\theta_j},
\end{align*}
where $\theta_j \in [0,2\pi)$, $t_j \in T$, and set
$v_j=e^{i\theta_j}$ in Lemma~\ref{lemma:int_t_dh_polynomial}. Again,
suppose $t_j=0$. Then
\begin{align}
 \abs{\int_{\SN^{\lambdalo}\brac{j} } \tau h\brac{\text{d} \tau}} & =\int_{\SN^{\lambdalo}\brac{j} } e^{-i\theta_j}  \tau h\brac{\text{d} \tau} \notag\\
&  \leq \abs{ \int_{\SN^{\lambdalo}\brac{j} } \brac{ q_1\brac{\tau}-e^{-i\theta_j} \tau  }h\brac{\text{d} \tau}}  +\int_{\SN^{\lambdalo}\brac{j} } q_1\brac{\tau} h\brac{\text{d} \tau}. \label{t_dh_bound_2}
\end{align}
The inequality~\eqref{q1_1} and H\"older's inequality allow to bound
the first term in the right-hand side of \eqref{t_dh_bound_2},
\begin{align}
 \abs{ \int_{\SN^{\lambdalo}\brac{j} } \brac{ q_1\brac{\tau}-e^{-i\theta_j} \tau  }h\brac{\text{d} \tau}}  & \leq  \int_{\SN^{\lambdalo}\brac{j} } \abs{ q_1\brac{\tau}-e^{-i\theta_j} \tau  }\abs{h}\brac{\text{d} \tau}\notag\\
& \leq \frac{C_a }{\lambdalo }\int_{\SN^{\lambdalo}\brac{j} }  \tau^2 \abs{h}\brac{\text{d} \tau} \notag \\
& \leq C_a \, \lambdalo \, I_{\SN^{\lambdalo} \brac{j}}\brac{h} . \label{t_dh_bound_3}
\end{align}
Another application of the triangular inequality yields
\begin{align}
\int_{\SN^{\lambdalo} } q_1\brac{\tau} h\brac{\text{d} \tau} & \leq \abs{\int_{\mathbb{T}} q_1\brac{\tau} h\brac{\text{d} \tau}} +\int_{\SF^{\lambdalo} } q_1\brac{\tau} h\brac{\text{d} \tau}. \label{t_dh_bound_4}
\end{align}
We employ H\"older's inequality, \eqref{eq:int_q}, \eqref{q1_1} and
\eqref{q1_2} to bound each of the terms in the right-hand side. First,
\begin{equation}
  \abs{\int_{\mathbb{T}} q_1\brac{\tau} h\brac{\text{d} \tau}} \leq \normInf{q_1} \normLOne{\Qlo h} \le C \, \lambdalo \, \delta. \label{t_dh_bound_5} 
\end{equation}
Second,
\begin{equation}
  \int_{\SF^{\lambdalo} } q_1\brac{\tau} h\brac{\text{d} \tau} \leq \normInf{P_{\SF^{\lambdalo}}\brac{q_1}} \normTV{P_{\SF^{\lambdalo}} \brac{h}} \leq C_b \, \lambdalo\, \normTV{P_{\SF^{\lambdalo}} \brac{h}}. \label{t_dh_bound_6}
\end{equation}
Combining \eqref{t_dh_bound_1} with these estimates gives 
\begin{align*}
 \sum_{t_j \in T} \abs{\int_{\SN^{\lambdahi}\brac{j} } \brac{\tau-t_j} h\brac{\text{d} \tau}} & \leq C \brac{ \lambdalo \, \delta + \lambdalo \, \normTV{P_{\SF^{\lambdalo}} \brac{h}} + \lambdalo \, I_{\SN^{\lambdalo} }\brac{h} +  \lambdahi \, \srf^2 \, I_{\SN^{\lambdalo} }\brac{h} },
\end{align*}
as desired. 

\section{Numerical implementation}
\label{sec:numerical}
In this section we discuss briefly how to solve problem~\eqref{TVproblem_relaxed_L2} by semidefinite programming. The dual problem of~\eqref{TVproblem_relaxed_L2} takes the form 
\begin{align*}
\max_{u \in \C^{n}} \; \operatorname{Re}\sqbr{\brac{\Flo \, y}^{\ast}u}-\delta \normTwo{u} \quad \text{subject to} \quad
& \normInf{\Flo^{\ast} \, u} \leq 1,
\end{align*}
where $\Flo$ denotes the linear operator that maps a
function to its first $n := 2\flo + 1$ Fourier coefficients as in
\eqref{eq:stochasticnoise} so that $\Qlo = \Flo^*
\Flo$. 
The dual can be recast as the semidefinite program (SDP)
\begin{align}
\label{TVdual}
\max_{u } \; \operatorname{Re}\sqbr{\brac{\Flo \, y}^{\ast}u}-\delta \normTwo{u} \quad \text{subject to} \quad &
\MAT{X & u \\ u^{\ast} & 1} \succeq 0, \notag\\
& \sum_{i=1}^{n-j}X_{i,i+j} = \begin{cases} 1, & j = 0,\\
     0, & j = 1, 2, \ldots, n-1, 
   \end{cases} 
\end{align}
where $X$ is an $n\times n$ Hermitian matrix, leveraging a corollary
to Theorem 4.24 in \cite{dumitrescu} (see also
\cite{atomic_norm_denoising,superres,cs_offgrid}).  In most cases,
this allows to solve the primal problem with high accuracy. The following lemma suggests how to obtain a primal solution from a dual solution.
\begin{lemma}
\label{lemma:primaldual}
Let $(\xestEq, \uestEq)$ be a primal-dual pair of solutions to
\eqref{TVproblem_relaxed_L2}--\eqref{TVdual}. For any $t \in \mathbb{T}$
with $\xestEq\brac{t}\neq
0$, 
\begin{align*}
 \brac{\FloEq^{\ast} \, \uestEq} \brac{t} = \signEq{\xestEq\brac{t}}.
\end{align*}
\end{lemma}
\begin{proof}
  First, we can assume that $y$ is low pass in the sense that $\QloEq
  y = y$. Since $\xest$ is feasible, $\|\Flo (y-\xest)\|_{\ell_2} =
  \|y- \Qlo \xest\|_{L_2} \leq \delta$.  Second, strong duality holds
  here. Hence, the Cauchy-Schwarz inequality gives
\[
\normTV{\xest}  =\operatorname{Re}\sqbr{\brac{\Flo \,
    y}^{\ast}\uest}-\delta \normTwo{\uest} =\PROD{\Flo\xest}{\uest}+
\PROD{\Flo y-\Flo\xest}{\uest}-\delta \normTwo{\uest} \leq
\PROD{\xest}{\Flo^{\ast} \uest }.
\]
By H\"older's inequality and the constraint on $\Flo^{\ast}\uest$,
$\normTV{\xest} \geq \PROD{\xest}{\Flo^{\ast} \uest} $ so that
equality holds. This is only possible if $\Flo^{\ast} \uest$ equals
the sign of $\xest$ at every point where $\xest$ is nonzero.
\end{proof}
This result implies that it is
usually possible to determine the support of the primal solution by
locating those points where the polynomial $q(t) = (\Flo^{\ast}
\uest)(t)$ has modulus equal to one. Once the support is estimated
accurately, a solution to the primal problem can be found by solving a
discrete problem. Figure~\ref{fig:sdp_example} shows the result of
applying this scheme to a simple example. We omit further details and
defer the analysis of this approach to future work.
\begin{figure}
\centering
\subfloat[][]{\includegraphics[width=10cm]{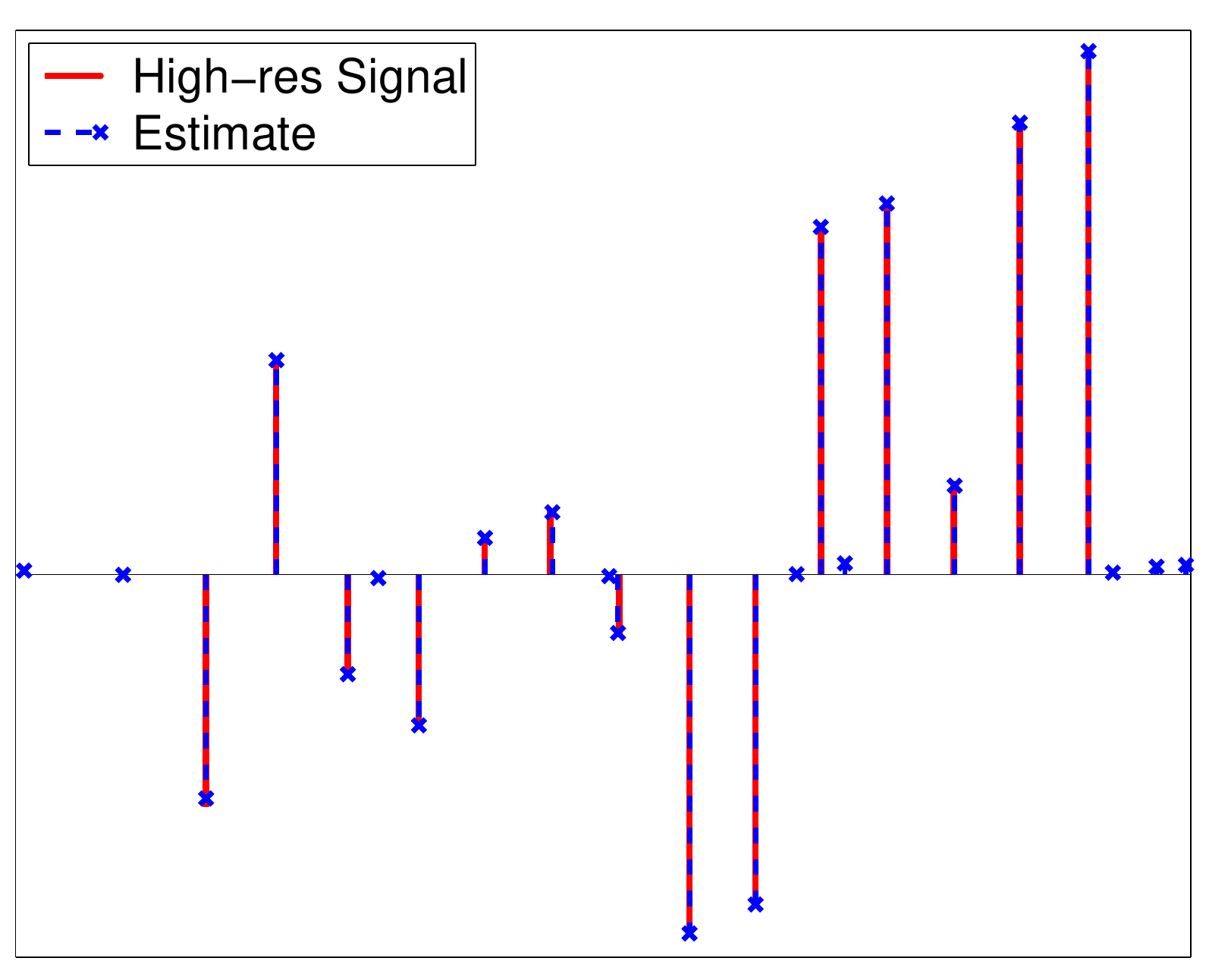}\label{subfig:lowres}}\\
\subfloat[][]{\includegraphics[width=10cm]{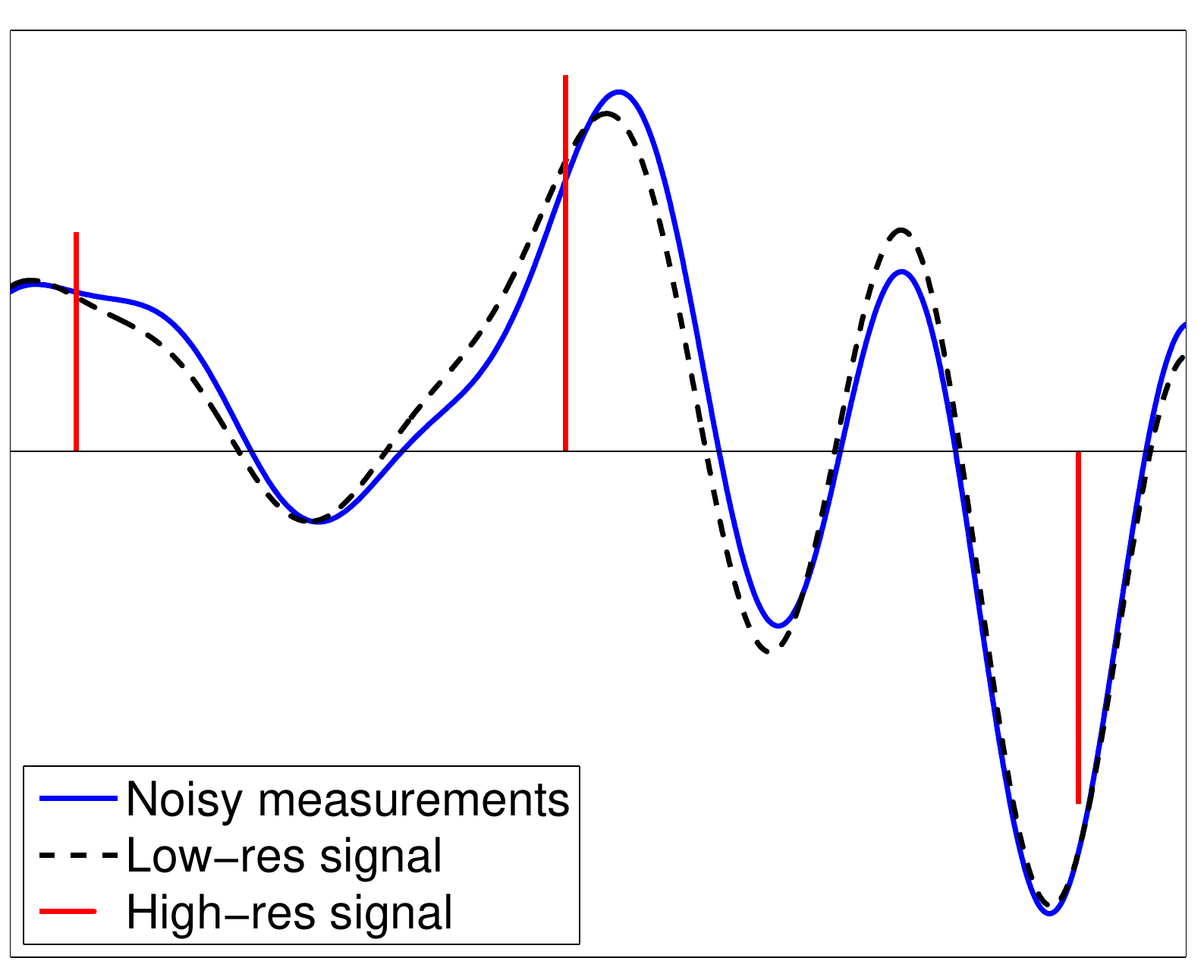}\label{subfig:recovery}}
\caption{  \protect\subref{subfig:lowres} Original signal (solid red spikes) and the estimate (dashed blue spikes with crosses) obtained by solving the sdp formulation of problem~\eqref{TVdual} using CVX~\cite{cvx}. The measurements consist of a noisy low-pass projection corresponding to the first 81 Fourier coefficients of the signal contaminated with i.i.d. Gaussian noise. The signal-to-noise ratio is 24.98 dB. Beyond the presence of some spurious spikes of small amplitude, the estimate is extremely precise.
 \protect\subref{subfig:recovery} Zoomed-in image of the noisy low-pass projection (solid blue line), the noiseless low-pass projection (dashed black line) and the original signal (red spikes).
}
\label{fig:sdp_example}
\end{figure}
\section{Discussion}
\label{sec:discussion}
In this work we introduce a theoretical framework that provides non-asymptotic stability guarantees for tractable super-resolution of multiple point sources in a continuous domain. More precisely, we show that it is possible to extrapolate the spectrum of a
superposition of point sources by convex programming and that the
extrapolation error scales quadratically with the super-resolution
factor. This is a worst case analysis since the noise has bounded norm
but is otherwise arbitrary. Natural extensions would include stability
studies using other error metrics and noise models. For instance, an
analysis tailored to a stochastic model might be able to sharpen
Corollary~\ref{cor:stochastic} and be more precise in its findings.
In a different direction, our techniques may be directly applicable to
related problems. An example concerns the use of the total-variation
norm for denoising line spectra~\cite{atomic_norm_denoising}. Here, it
would be interesting to see whether our methods allow to prove better
denoising performance under a minimum-separation condition. Another
example concerns the recovery of sparse signals from a random subset
of their low-pass Fourier coefficients~\cite{cs_offgrid}. Here, it is
likely that our work would yield stability guarantees from noisy
low-frequency data.

\subsection*{Acknowledgements}

E. C. is partially supported by AFOSR under grant FA9550-09-1-0643, by
ONR under grant N00014-09-1-0258 and by a gift from the Broadcom
Foundation. C.~F.~is supported by a Fundaci\'on Caja Madrid
Fellowship. We thank Carlos Sing-Long for useful feedback about an earlier version of the manuscript.

\bibliographystyle{abbrv}
\bibliography{refs_noise}

\clearpage
\appendix

\section{Proof of Lemma~\ref{lemma:dualpol_extra}}
\label{sec:dualpol_extra}

We use the construction described in
Section 2 of~\cite{superres}. In more detail,
\begin{equation*}
  q(t)  = \sum_{t_k \in T} \alpha_k \Ker(t-t_k) + \beta_k \Ker^{\brac{1}}(t-t_k), 
\end{equation*}
where $\alpha, \beta \in \C^{\abs{T}}$ are coefficient vectors,
\begin{equation} 
  \Ker(t) = \left[\frac{\sin \brac{\brac{\frac{\flo}{2}+1} \pi
        t}}{\brac{\frac{\flo}{2}+1}\sin \brac{\pi t}}\right]^4, \quad
  t \in \mathbb{T} \setminus\{0\}, 
\label{def:kernel}
\end{equation}
and $\Ker(0) = 1$; here, $\Ker^{\brac{\ell}}$ is the $\ell$th
derivative of $\Ker$. If $\flo$ is even, $\Ker(t)$ is the square of
the Fej\'er kernel. By construction, the coefficients $\alpha$ and
$\beta$ are selected such that for all $t_j \in T$,
\begin{align*}
  q(t_j) & = v_j\\
  q'(t_j) & = 0.
\end{align*}
Without loss of generality we consider $t_j=0$ and bound
$q\brac{t}-v_j$ in the interval $\sqbr{0,0.16\lambdalo}$. To ease
notation, we define $w(t)=q\brac{t}-v_j=w_R(t)+i\,w_I(t)$, where $w_R$
is the real part of $w$ and $w_I$ the imaginary part. Leveraging
different results from Section 2 in~\cite{superres} (in particular
the equations in (2.25) and Lemmas 2.2 and 2.7), we have
\begin{align*}
  \abs{w_R''\brac{t}} & =  \abs{\sum_{t_k \in T}\Real{\alpha_k} \Ker^{\brac{2}}\brac{t-t_k} + \sum_{t_k \in T}\Real{\beta_k} \Ker^{\brac{3}}\brac{t-t_k}}\notag\\
  & \leq \normInf{\alpha}\sum_{t_k \in T}\abs{\Ker^{\brac{2}}\brac{t-t_k}} + \normInf{\beta}\sum_{t_k \in T} \abs{\Ker^{\brac{3}}\brac{t-t_k}}\notag\\
& \leq C_{\alpha } \brac{\abs{\Ker^{\brac{2}}\brac{t}} +  \sum_{t_k \in T\setminus \{0\}} \abs{\Ker^{\brac{2}}\brac{t-t_k}} } + C_{\beta}\lambdalo \brac{\abs{\Ker^{\brac{3}}\brac{t}}+  \sum_{t_k \in T\setminus \{0\}} \abs{\Ker^{\brac{3}}\brac{t-t_k}}}\\
& \leq C \, \flo^2.
\end{align*}
The same bound holds for $w_I$. Since $w_R(0)$, $w_R'(0)$, $w_I(0)$
and $w_I'(0)$ are all equal to zero, this implies $\abs{w_R(t)}\leq C'
\flo^2 t^2$ and $\abs{w_I(t)}\leq C' \flo^2 t^2$ in the interval of
interest, which allows the conclusion
\begin{align*}
\abs{w(t)} \leq C \, \flo^2 t^2.
\end{align*}

\section{Proof of Lemma \ref{lemma:int_t_dh_polynomial}}
\label{sec:int_t_dh_polynomial}
The proof is similar to that of Lemma \ref{lemma:dualpol} (see Section
2 of \cite{superres}), where a low-frequency kernel and its derivative
are used to interpolate an arbitrary sign pattern on a support
satisfying the minimum-distance condition. More precisely, we set
\begin{equation}
  q_1(t)  = \sum_{t_k \in T} \alpha_k \Ker(t-t_k) + \beta_k \Ker^{\brac{1}}(t-t_k), 
  \label{def:q_1}
\end{equation}
where $\alpha, \beta \in \C^{\abs{T}}$ are coefficient vectors, $\Ker$ is
defined by~\eqref{def:kernel}. Note that $\Ker$, $\Ker^{\brac{1}}$ and,
consequently, $q_1$ are trigonometric polynomials of degree at most
$f_0$.  By Lemma 2.7 in \cite{superres}, it holds that for any $t_0
\in T$ and $t \in \mathbb{T}$ obeying $\abs{t-t_0}\leq 0.16
\lambdalo$,
\begin{align}
  \sum_{t_k \in T\setminus \{t_0\}} \abs{\Ker^{\brac{\ell}}\brac{t-t_k}}
  \leq C_{\ell} \flo^{\ell},
\label{boundF}
\end{align}
where $C_{\ell}$ is a positive constant for $\ell=0,1,2,3$; in
particular, $C_0 \leq 0.007$, $C_1 \leq 0.08$ and $C_2 \leq 1.06$. In
addition, there exist other positive constants $C_0'$ and $C_1'$, such
that for all $t_0 \in T$ and $t\in \mathbb{T}$ with $\abs{t-t_0}\leq
\Delta/2$,
\begin{align}
  \sum_{t_k \in T\setminus \{t_0\}} \abs{\Ker^{\brac{\ell}}\brac{t-t_k}}
  \leq C_{\ell}' \flo^{\ell}
\label{boundF2}
\end{align}
for $\ell=0,1$. We refer to Section 2.3 in \cite{superres} for a
detailed description of how to compute these bounds.

In order to satisfy \eqref{q1_1} and \eqref{q1_2}, we constrain $q_1$
as follows: for each $t_j \in T$, 
\begin{align*}
q_1(t_j) & = 0,\\
q_1'(t_j) & = v_j.
\end{align*}
Intuitively, this forces $q_1$ to approximate the linear function
$v_j\brac{t-t_j}$ around $t_j$. These constraints can be expressed in
matrix form,
\[
\begin{bmatrix} D_0 & D_1\\ D_1 & D_2 \end{bmatrix} \begin{bmatrix} \alpha \\
  \beta \end{bmatrix} =\begin{bmatrix} 0\\  v \end{bmatrix},
\]
where
\[
\brac{D_0}_{jk} = \Ker\brac{t_j - t_k}, \quad \brac{D_1}_{jk} =
\Ker^{\brac{1}}\brac{t_j - t_k}, \quad \brac{D_2}_{jk} = \Ker^{\brac{2}}\brac{t_j - t_k},
\]
and $j$ and $k$ range from $1$ to $\abs{T}$. It is shown in Section 2.3.1 of \cite{superres} that under the minimum-separation condition this system is invertible, so that $\alpha$ and $\beta$ are well defined. These coefficient vectors can consequently be expressed as
\[
  \begin{bmatrix} \alpha \\
    \beta \end{bmatrix} =\begin{bmatrix} -D_0^{-1}D_1 \\
    \Id \end{bmatrix} S^{-1} v, \quad S :=
  D_2-D_1D_0^{-1}D_1,
\]
where $S$ is the Schur complement. Inequality~\eqref{boundF} implies
\begin{align}
\normInfInf{\Id- D_0} & \le
  C_0  \label{Id_D_bound},\\
\normInfInf{D_1}   & \le C_{1} \flo, \label{D1_bound}\\
\normInfInf{\kappa\Id-D_2} & \le 
  C_{2} \flo^2
\label{Id_D2_bound}\text{,}
\end{align} 
where $\kappa=\abs{\Ker^{\brac{2}}(0)} =\pi^2 \flo (\flo + 4)/3$.\\

Let $\|M \|_\infty$ denote the usual infinity norm of a matrix $M$ defined as
$\|M\|_\infty = \max_{\|x\|_\infty = 1} \|Mx\|_\infty = \max_i \sum_j
|a_{ij}|$. Then, if $\normInfInf{I - M} < 1$, the series $M^{-1} = \brac{I - \brac{I - M}}^{-1} = \sum_{k \ge 0}
\brac{I - M}^k$ is convergent and we have
\begin{equation*}
  \normInfInf{M^{-1}}  \leq \frac{1}{1-\normInfInf{\Id-M}}.
\end{equation*}
This, together with \eqref{Id_D_bound}, \eqref{D1_bound} and \eqref{Id_D2_bound} implies
\begin{align*}
\normInfInf{D_0^{-1}} & \leq \frac{1}{1-\normInfInf{\Id-D_0}} \leq \frac{1}{1-C_0} ,\\ 
\normInfInf{\kappa \Id - S} & \leq  \normInfInf{\kappa \Id-D_2}+\normInfInf{D_1}\normInfInf{D_0^{-1}}\normInfInf{D_1} \leq \brac{C_2+\frac{C_1^2}{1-C_0}} \flo^2 , \\
\normInfInf{S^{-1}} & = \kappa^{-1}\normInfInf{\brac{\frac{S}{\kappa}}^{-1}} \leq \frac{1}{\kappa-\normInfInf{\kappa\Id-S}} \leq \brac{\kappa-\brac{C_2+\frac{C_1^2}{1-C_0}} \flo^2}^{-1} \leq C_{\kappa} \lambdalo^{2},
\end{align*}
for a certain positive constant $ C_{\kappa}$. Note that due to the
numeric upper bounds on the constants in \eqref{boundF} $C_{\kappa}$
is indeed a positive constant as long as $\flo \geq 1$. Finally, we
obtain a bound on the magnitude of the entries of $\alpha$
\begin{align}
\normInfInf{\alpha} & = \normInfInf{D_0^{-1}D_1 S^{-1} v} \leq  \normInfInf{D_0^{-1}D_1 S^{-1}} \leq \normInfInf{D_0^{-1}} \, \normInfInf{D_1}  \, \normInfInf{S^{-1}} \leq C_{\alpha} \lambdalo,\label{boundAlpha}
\end{align}
where $C_{\alpha} =C_{\kappa} C_1 /\brac{1-C_0}$, and on the entries of  $\beta$
\begin{equation}
\normInfInf{\beta}=\normInfInf{S^{-1} v}\leq \normInfInf{S^{-1}}  \leq C_{\beta}\lambdalo^{2},\label{boundBeta} 
\end{equation}
for a positive constant $C_{\beta}=C_{\kappa}$. Combining these inequalities with \eqref{boundF2} and the fact that the absolute values of $\Ker(t)$ and $\Ker^{\brac{1}}(t)$ are bounded by one and $7 \flo$ respectively (see the proof of Lemma C.5 in \cite{superres}), we have that for any $t$
\begin{align}
 |q_1(t)| & = \abs{\sum_{t_k \in T}\alpha_k \Ker\brac{t-t_k} + \sum_{t_k \in T} \beta_k \Ker^{\brac{1}}\brac{t-t_k}}\notag\\
  & \leq \normInfInf{\alpha}\sum_{t_k \in T}\abs{\Ker\brac{t-t_k}} + \normInfInf{\beta}\sum_{t_k \in T} \abs{\Ker^{\brac{1}}\brac{t-t_k}}\notag\\
& \leq C_{\alpha } \lambdalo \brac{\abs{\Ker\brac{t}} +  \sum_{t_k \in T\setminus \{t_i\}} \abs{\Ker\brac{t-t_k}} } + C_{\beta}\lambdalo^2 \brac{\abs{\Ker^{\brac{1}}\brac{t}}+  \sum_{t_k \in T\setminus \{t_i\}} \abs{\Ker^{\brac{1}}\brac{t-t_k}}} \notag\\
& \leq C \lambdalo,
\end{align}
where $t_i$ denotes the element in $T$ nearest to $t$ (note that all
other elements are at least $\Delta/2$ away). Thus, \eqref{q1_2}
holds.

The proof is completed by the following lemma, which proves 
\eqref{q1_1}.
\begin{lemma}
\label{lemma:concavity}
For any $t_j \in T$ and $t \in \mathbb{T}$ obeying $\abs{t-t_j} \leq
0.16 \lambdaloEq$, we have 
\begin{align*}
|q_1(t)-v_j\brac{t-t_j}| & \leq \frac{C \brac{t-t_j}^2}{\lambdaloEq}. 
\end{align*}
\end{lemma}
\begin{proof}
We assume without loss of generality that $t_j = 0$. By symmetry, it suffices to show the claim for $t \in (0,0.16 \,
\lambdalo]$. To ease notation, we define $w(t)=v_j t-q_1(t)=w_R(t)+i\,w_I(t)$, where $w_R$ is the real part of $w$ and $w_I$ the imaginary part. Leveraging \eqref{boundAlpha}, \eqref{boundBeta} and \eqref{boundF} together with the fact that $\Ker^{\brac{2}}(t)$ and $\Ker^{\brac{3}}(t)$ are bounded by $4 \flo^2$ and $6 \flo^3$ respectively if $\abs{t}\leq 0.16 \lambdalo$ (see the proof of Lemma 2.3 in \cite{superres}), we obtain
\begin{align*}
  \abs{w_R''\brac{t}} & =  \abs{\sum_{t_k \in T}\Real{\alpha_k} \Ker^{\brac{2}}\brac{t-t_k} + \sum_{t_k \in T}\Real{\beta_k} \Ker^{\brac{3}}\brac{t-t_k}}\notag\\
  & \leq \normInfInf{\alpha}\sum_{t_k \in T}\abs{\Ker^{\brac{2}}\brac{t-t_k}} + \normInfInf{\beta}\sum_{t_k \in T} \abs{\Ker^{\brac{3}}\brac{t-t_k}}\notag\\
& \leq C_{\alpha } \lambdalo \brac{\abs{\Ker^{\brac{2}}\brac{t}} +  \sum_{t_k \in T\setminus \{0\}} \abs{\Ker^{\brac{2}}\brac{t-t_k}} } + C_{\beta}\lambdalo^2 \brac{\abs{\Ker^{\brac{3}}\brac{t}}+  \sum_{t_k \in T\setminus \{0\}} \abs{\Ker^{\brac{3}}\brac{t-t_k}}}\\
& \leq C \, \flo. 
\end{align*}
The same bound applies to $w_I$. Since $w_R(0)$, $w_R'(0)$, $w_I(0)$ and
$w_I'(0)$ are all equal to zero, this implies $\abs{w_R(t)}\leq C \flo
t^2$---and similarly for $\abs{w_I(t)}$---in the interval of interest. Whence, 
$\abs{w(t)} \leq C \flo t^2$.
\end{proof}

\section{Proof of Corollary~\ref{cor:stochastic}}
\label{sec:stochastic}
The proof of Theorem~\ref{theorem:noise} relies on two identities
\begin{align}
\normTV{\xest} & \leq \normTV{x}, \label{eq:ineq1}\\
\normLOne{\Qlo \brac{ \xest- x} } & \leq 2\delta, \label{eq:ineq2}
\end{align}
which suffice to establish 
\begin{equation*}
  \normLOne{\Khi*(\xest- x)} \leq C_0 \, \srf^2 \, \delta.  
\end{equation*}
To prove the corollary, we show that~\eqref{eq:ineq1} and~\eqref{eq:ineq2} hold. Due to the fact that $\normTwo{\epsilon}^2$ follows a $\chi^2$-distribution with $4\flo+2$ degrees of freedom, we have  
\begin{equation*}
  \P \brac{ \normTwo{\epsilon} > \brac{1+\gamma}\sigma \sqrt{4 \flo +2}=\delta} < e^{- 2\flo \gamma^2 },
\end{equation*}
for any positive $\gamma$ by  a concentration inequality (see~\cite[Section
4]{chisquare}). By Parseval, this implies that with high probability $\normLTwo{\Qlo x -y}=\normTwo{\epsilon}\leq \delta$. As a result, $\xest$ is feasible, which implies~\eqref{eq:ineq1} and furthermore
\begin{align*}
  \normLOne{\Qlo \brac{ \xest- x}} \leq \normLTwo{\Qlo \brac{ \xest- x}} \leq \normLTwo{\Qlo x -y}+\normLTwo{y-\Qlo \xest} \leq 2\delta,
\end{align*}
since by the Cauchy-Schwarz inequality $\normLOne{f} \leq \normLTwo{f}$ for any function $f$ with bounded $L_2$ norm supported on the unit interval. Thus, \eqref{eq:ineq2} also holds and the proof is complete.

\section{Extension to multiple dimensions}
\label{sec:multidim}

The extension of the proof hinges on establishing versions of Lemmas~\ref{lemma:dualpol},~\ref{lemma:dualpol_extra} and~\ref{lemma:int_t_dh_polynomial} for multiple dimensions.
These lemmas construct bounded low-frequency polynomials which interpolate a sign pattern on a well-separated set of points $S$ and have bounded second
derivatives in a neighborhood of $S$. In the multidimensional case, we need the directional derivative of the polynomials to be bounded in any direction, which can be ensured by bounding the eigenvalues of their Hessian matrix evaluated on the support of the signal. 
To construct such polynomials one can proceed in a way similar
to the proof of Lemmas~\ref{lemma:dualpol}
and~\ref{lemma:int_t_dh_polynomial}, namely, by using a low-frequency
kernel constructed by tensorizing several squared Fej\'er kernels to
interpolate the sign pattern, while constraining the first-order
derivatives to either vanish or have a fixed value. As in the one-dimensional case, one can
set up a system of equations and prove that it is well conditioned
using the rapid decay of the interpolation kernel away from the
origin. Finally, one can verify that the construction satisfies the
required conditions by exploiting the fact that the interpolation
kernel and its derivatives are locally quadratic and rapidly
decaying. This is spelled out in the proof of Proposition~C.1
in~\cite{superres} to prove a version of Lemma~\ref{lemma:dualpol} in
two dimensions. In order to clarify further how to adapt our techniques to a multidimensional setting we provide below a sketch of the proof of the analog of Lemma~\ref{lemma:TC_bound} in two dimensions. In particular, this illustrates how the increase in dimension does not change the exponent of the SRF in our recovery guarantees. 

\subsection{Proof of an extension of Lemma~\ref{lemma:TC_bound} to two dimensions}

We now have $t \in \mathbb{T}^2$. As a result, we redefine
\begin{align*}
  \SN^{\lambda} \brac{j} & := \keys{ t  \; : \; \normInf{t-t_j}\leq w \, \lambda},\\
  \SF^{\lambda} &:= \keys{ t \; : \; \normInf{t-t_j}> w \, \lambda, \;
    \forall t_j \in T},\\
  I_{\SN^{\lambda} \brac{j}}\brac{\mu} &:=
  \frac{1}{\lambdalo^2}\int_{\SN^{\lambda} \brac{j}} \normTwo{t-t_j}^2
  \abs{\mu} \brac{\text{d} t},
 \end{align*}
 where $w$ is a constant. 
 
The proof relies on the existence of a low-frequency
polynomial 
\begin{equation*}
q(t) = \sum_{k_1 = -\floEq}^{\floEq}\sum_{k_2 = -\floEq}^{\floEq} c_{k_1,k_2} e^{i2\pi (k_1 t_1 +k_2 t_2)} 
\end{equation*}
satisfying
\begin{align}
q(t_j) & = v_j, \quad t_j \in T, \label{q0_1_2D}\\
|q(t)| & \leq 1-\frac{C'_a  \normTwo{t-t_j}^2}{\lambdaloEq^2} ,  \quad  t \in \SNEq^{\lambdaloEq}\brac{j}, \label{q0_2_2D}\\
|q(t)| & < 1-C'_b ,  \quad  t \in \SFEq^{\lambdaloEq}, \label{q0_3_2D}
\end{align}
where $C_a'$ and $C_b'$ are constants. Proposition~C.1 in~\cite{superres} constructs such a polynomial. Under a minimum distance condition, which constrains the elements of $T$ to be separated by $2.38 \, \lambdalo$ in infinity norm (as explained in~\cite{superres} this choice of norm is arbitrary and could be changed to the $\ell_2$ norm), \cite{superres} shows that $q$ satisfies \eqref{q0_1_2D} and \eqref{q0_3_2D} and that both eigenvalues of its Hessian matrix evaluated on $T$ are of order $\flo^2$, which implies $\eqref{q0_2_2D}$.

As in one dimension, we perform a polar decomposition of $\PTn h$,
\begin{align*}
 \PTn h = e^{i \phi \brac{t}}\abs{\PTn h}\text{,}
\end{align*}
and work with $v_j = e^{-i \phi(t_j)}$. The rest of the proof is almost identical to the 1D case. Since
$q$ is low frequency,
\begin{align}
\abs{\int_{\mathbb{T}^2} q(t) \text{d}h (t)} &  \leq 2 \delta. \label{eq:int_q_2D}
\end{align}
Next, since $q$ interpolates $e^{-i \phi \brac{t}}$ on $T$,
\begin{align}
  \normTV{ \PTn h } = \int_{\mathbb{T}^2} q(t) \PTn h\brac{\text{d}t} &\leq 2\delta + \sum_{j\in T}
  \abs{\int_{\SN^{\lambdalo}\brac{j}\setminus\keys{t_j} } q(t)
    h\brac{\text{d}t}}+ \abs{\int_{\SF^{\lambdalo}} q(t)
    h\brac{\text{d}t}}. \label{PTh_1c_2D}
\end{align}
Applying \eqref{q0_3_2D} and H\"older's
inequality, we obtain
\begin{align}
\abs{\int_{\SF^{\lambdalo}}  q(t) h\brac{\text{d}t}} &  \leq \brac{1-C_b'}\normTV{P_{\SF^{\lambdalo}}\brac{h}}.  \label{PTh_2c_2D}
\end{align}
Setting $t_j=(0,0)$ without loss of generality, the triangle inequality and
\eqref{q0_2_2D} yield
\begin{align}
\abs{\int_{\SN^{\lambdalo}\brac{j}\setminus\keys{(0,0)}}  q(t) h\brac{\text{d}t}} 
& \leq \int_{\SN^{\lambdalo}\brac{j}\setminus\keys{(0,0)}} \abs{h}\brac{\text{d}t} - C_a' I_{\SN^{\lambdalo} \brac{j}}\brac{h} . \label{PTh_3c_2D}
\end{align}
Combining \eqref{PTh_1c_2D}, \eqref{PTh_2c_2D} and \eqref{PTh_3c_2D} gives 
\begin{align*}
\normTV{ \PTn h } \leq &  2\delta + \normTV{ \PTcn h } - C_b'  \normTV{P_{\SF^{\lambdalo}}\brac{h}}-  C_a' I_{\SN^{\lambdalo}}\brac{h}
\end{align*}
and similarly
\[
\normTV{ \PTn h } \leq   2\delta + \normTV{ \PTcn h } - w^2 \, C_a' \, \srf^{-2}  \normTV{P_{\SF^{\lambdahi}}\brac{h}}-  C_a' I_{\SN^{\lambdahi}}\brac{h}.
\]
By the same argument as in the 1D case, the fact that $\hat{x}$ has minimal total-variation norm is now sufficient to establish
\begin{align*}
C_b' \normTV{P_{\SF^{\lambdalo}}\brac{h}} + C_a' I_{\SN^{\lambdalo}}\brac{h} \leq  2 \delta, 
\end{align*}
and
\begin{align*}
 w^2 \, C_a' \, \srf^{-2} \normTV{P_{\SF^{\lambdahi}}\brac{h}} + C_a' I_{\SN^{\lambdahi}}\brac{h} \leq  2 \delta.
\end{align*}

\end{document}